\documentclass[eqsecnum,floats,aps,nofootinbib,showpacs]{revtex4}
\usepackage{latexsym}
\usepackage{amssymb}

\usepackage[T1]{fontenc}
\usepackage{bm,exscale,latexsym, amsmath, amssymb, mathrsfs, dsfont,accents,tensor} 
\usepackage[dvips]{graphicx}

\usepackage{hyperref}

\usepackage{amsmath, amsthm, amssymb}
\newtheorem{theorem}{Theorem}[section]
  
\newtheorem{lemma}[theorem]{Lemma}

\def\be{\begin{equation}}
\def\ee{\end{equation}}
\def\beq{\begin{equation*}}
\def\eeq{\end{equation*}}
\def\ba{\begin{aligned}}
\def\ea{\end{aligned}}

\def\ov{\overline}

\begin{document}
\title{Generalized Dirac bracket and the role of the Poincar\'e symmetry \break in the program of canonical quantization of fields 1}

\author {Marcin Ka\'zmierczak}
\email{marcin.kazmierczak@fuw.edu.pl}
\affiliation{Institute of Theoretical Physics, University of Warsaw
  ul. Ho\.{z}a 69, 00-681 Warszawa, Poland}

\begin{abstract}
An elementary presentation of the methods for the canonical quantization of constraint systems with Fermi variables is given. The emphasis is on the subtleties of the construction of an appropriate classical bracket that could be consistently replaced by commutators or anti--commutators of operators, as required by canonical quantization procedure for bosonic and fermionic degrees of freedom respectively. I present a consequent canonical quantization of the Dirac field, in which the role of Poincar\'e invariance is made marginal. This simple example provides an introduction to the Poincar\'e--free quantization of spinor electrodynamics in the second part of the paper. 
\end{abstract}
 

\maketitle

\section{Introduction}\label{intr}

The canonical quantization scheme can be briefly described as follows: to find a quantum version of a given classical theory, one formulates this classical theory in the Hamiltonian framework in terms of Poisson brackets. Then the classical dynamical variables are promoted to the operators and the Poisson brackets to the commutators (up to the factor $i\hbar$). The space of states is obtained by looking for representations of the commutation relations of physically important observables in a Hilbert space.

The usefulness of these method of quantization was evident from the very beginning of quantum theory. The historic monograph by Dirac \cite{Dirac1} contains a beautiful presentation of the underlying motivations. With the realization that the theories of particles, both massless and massive, need to be viewed as quantum field theories, the attempts of canonical quantization of field--theoretical systems began.
 Although the extension of the formalism to the uncountable number of degrees of freedom did not cause much trouble (except for the usual difficulties in mathematically rigorous formulation), there were two other features of physically relevant field theories that were problematic. 

One of them was related to the occurrence of degenerate Lagrangians that lead to the presence of constraints in the Hamiltonian formalism. These problems were partially resolved already by Dirac \cite{Dirac2}. It appeared that the constraints may be of two kinds. One of them leads to the presence of gauge freedom in a physical system and the other to the necessity of replacing the conventional Poisson bracket by the new classical bracket (the Dirac bracket) in the quantization procedure.

 The second problem was bound up with the relation between spin and statistics. As explained in classical refferences \cite{Wein}, the field operators for particles with half--integer values of spin that are evaluated at spatially separated space--time points should anti--commute, rather then commute. This poses a problem for the canonical quantization scheme, since both the Poisson bracket and the Dirac bracket are antisymmetric in their arguments and hence cannot be consistently replaced by symmetric anti--commutators of operators. The solution to this problem was proposed by Bellinfante {\it at all.} \cite{Bell}. The methods discussed there acquired a rigorous mathematical formulation in the papers by 
Casalbouri \cite{Casal}\cite{Casal2}. It appeared that it was necessary to introduce two kinds of classical variables, whose multiplication is not necessarily commutative. Specifically, the multiplication of the so called even--type variables with all the others is commutative and the multiplication of odd--type ones between themselves is anti--commutative. After these Grassman variables are introduced, it is possible to define the generalized Poisson bracket that is symmetric whenever both variables are odd and anti--symmetric otherwise. Also, the bracket possesses other important algebraic properties. The introduction of Grassman variables in the presence of the constraints is discussed in \cite{Casal2}.

More recently, the discussion of constrained systems with Fermi variables can be found in \cite{HT} and \cite{JHP}, but the authors of these references decided to focus on general considerations, rather than the applications of the formalism they developed. Handbooks of quantum field theory either restrain from discussing the canonical quantization of constrained systems at all \cite{Peskin} in favor of path integral approach, or present the discussion of the Dirac bracket that is relevant for bosonic systems only \cite{Wein}. The anti--commutation relations for the Dirac field are then derived from
the abstract group theoretical arguments that invoke to the Poincar\'e symmetry of space--time and the assumed Lorentz transformation properties of the Dirac field, as well as the discrete symmetries and the causality arguments. The canonical anti--commutation relations are then not necessary. Although this approach is indisputably elegant, it has a failure of not being expendable to the case of curved space--time, which does not posses the Poincar\'e symmetry. On the other hand,
the advantage of the canonical quantization is that the Poincar\'e symmetry does not have to be employed at all. Not only the space--time metric does not need to be flat -- it is not necessary at all in the quantization procedure.
Indeed, much of the motivations 
underlying the development of canonical methods was related to the attempts to quantize gravity in a background independent way.
See e.g. \cite{Thiemann} for the presentation of one of the advanced manifestations of these attempts. One of the points raised in this reference was the prevalence of the Poincar\'e symmetry and the frequent usage of the Minkowski metric in standard quantum field theory (see p.3 of \cite{Thiemann}), e.g. when formulated in terms of Wightman axiomatics. The author finds these properties of QFT an important obstacle to the straightforward extension to the gravitational case. One of the aims of my article is to show that the Poincar\'e symmetry and the Minkowski metric need not be used {\it almost} at all, if the canonical quantization scheme is consequently followed (the meaning of {\it almost} will become clear during the presentation). In this first paper of the series I will review the general formalism for canonical analysis of constraint systems with Fermi degrees of freedom. Than I will focus on the theory of the free Dirac field, which is the simplest physically relevant example containing both the difficulties that were mentioned.  In the forthcoming article the electromagnetic interaction will be turned on and then the issue of gauge invariance and gauge fixing will be addressed in more details.

\section{The canonical analysis of constrained systems with finite number of degrees of freedom}\label{general}

\subsection{Classical mechanics of constrained systems}

Our starting point is a classical theory in a Lagrangian formulation, whose equations of motion are defined by the stationarity condition for the action functional
\be\label{SL}
S[q(t)]=\int_{t_1}^{t_2} L(q,\dot q)dt \quad ,
\ee
where $q=(q^1,\dots q^N)$ represents the positions. In order to pass to Hamiltonian formulation, one defines canonical momenta as functions of positions and velocities
\be\label{momenta}
p_n:=\frac{\partial L}{\partial \dot q^n}(q,\dot q) , \quad n=1\dots N
\ee
and the canonical Hamiltonian
\be\label{H}
H:=p_n \dot q^n-L
\ee
(here and further in the article the Einstein summation convention applies whenever the indexes repeat, unless otherwise stated). Although $H$ is originally a function of positions and velocities on account of (\ref{momenta}), it can be expressed  as a function of positions and momenta instead, since its variation 
\be
\delta H=\left({p_n-\frac{\partial L}{\partial \dot q^n}}\right)\delta \dot q^n + \dot q^n \delta p_n - \frac{\partial L}{\partial q^n} \delta q^n = \dot q^n \delta p_n - \frac{\partial L}{\partial q^n} \delta q^n
\ee
depends on $\delta\dot q$ only through the combinations $\delta p_n=\frac{\partial^2 L}{\partial q^{n'}\partial \dot q^n}\delta q^{n'}+\frac{\partial^2 L}{\partial \dot q^{n'}\partial \dot q^n}\delta \dot q^{n'}$. An expression for $H$ as a function of the $q$'s and the $p$'s, although always exists, needs not be unique. To find one in practice, one could wish to calculate the velocities as functions of positions and momenta from (\ref{momenta}) and insert the result into (\ref{H}). However, the condition for this to be possible for all the $\dot q$'s is that the determinant of the matrix $\partial^2 L/\partial\dot q^{n'}\partial\dot q^n$ does not vanish. If this condition fails to hold, which is the case for many systems of direct physical importance, then one cannot compute form (\ref{momenta}) {\it all} of the $\dot q$'s as functions of the $q$'s and the $p$'s. This is however not necessary, since in this case (\ref{momenta}) implies some relations between the $q$'s and the $p$'s of the form
\be\label{primary}
\phi_m(q,p)=0, \quad m=1,\dots,M\quad 
\ee
which make the dependence of $H$ on the remaining velocities vanish.
These conditions, called the {\it primary constraints}, define a submanifold in phase space, called the {\it primary constraint surface}. The canonical Hamiltonian is well defined only as a function on this surface. It is therefore allowed to use (\ref{primary}) when expressing $H$ as a function of the $q$'s and the $p$'s, thus obtaining many equivalent expressions for $H$.

The Hamilton equations of motion that are equivalent to the Euler--Lagrange equations of (\ref{SL}) are given by
\be\label{HamEq}
\begin{aligned}
&\dot q^n=\frac{\partial H}{\partial p_n}+u^m\frac{\partial \phi_m}{\partial p_n} ,\\
&\dot p_n=-\frac{\partial H}{\partial q^n}-u^m \frac{\partial \phi_m}{\partial q^n} , \\
&\phi_m=0 ,
\end{aligned}
\ee
where the $u$'s are functions on phase space whose dependence on the $q$'s and the $p$'s should be determined in such a way that the system of equations (\ref{HamEq}) have a solution (if this is not possible than it means that the system of Euler--Lagrange equations of (\ref{SL}) was contradictory) . See  \cite{HT} for the proof of the equivalence of (\ref{HamEq}) and the Euler--Lagrange equations. For our purposes it is important to note that the Hamilton equations can be rewritten as
\be\label{HamEq1}
\dot q^n=[q^n,H_T]_{P} ,\quad \dot p_n=[p_n,H_T]_{P} ,\quad \phi_m = 0 ,
\ee
where 
\be\label{H_T}
H_T:=H+u^m\phi_m
\ee
is called the {\it total Hamiltonian} and $[,]_{P}$ stands for the usual Poisson bracket (PB) which is defined for any dynamical variables $F(q,p)$ and $G(q,p)$ by
\be
[F,G]_{P}=\frac{\partial F}{\partial q^n}\frac{\partial G}{\partial p_n}-\frac{\partial G}{\partial q^n}\frac{\partial F}{\partial p_n} .
\ee
Note that it is not necessary to calculate the PB between $u$'s and the $q$'s and the $p$'s in (\ref{HamEq1}), since all the brackets containing $u$'s will be multiplied by the constraints, and thus this components will vanish on account of the last equation of (\ref{HamEq1}). The evolution equations for any dynamical variable $F(q,p)$ can now be expressed in an extremely simple form
\be\label{dotFH_T}
\dot F\approx [F,H_T]_{P} , 
\ee 
where the {\it weak equality} symbol $\approx$ means that the equality holds if the constraints are imposed on the final form of the expressions (note that the constraints cannot be imposed on $H_T$ before its Poisson bracket with $F$ is calculated!).  

In order to quantize the theory, one needs to be aware of {\it all} the constraints, i. e. it is necessary to find all the independent relations of the form $C(p,q)=0$ that are satisfied at any instant of time. In general, the primary constraints $\phi_m$ will not provide a complete set of such relations. Additional constraints can follow from the requirement that $\phi_m$'s are preserved in time:
\be\label{constrpres}
\dot \phi_m\approx [\phi_m,H_T]_{P}\approx [\phi_m,H]_P+u^{m'}[\phi_m,\phi_{m'}]_{P}\approx 0 .
\ee
For e given $m$, this equation can provide a restriction on $u$'s (which happens whenever there exists a constraint that does not commute with $\phi_m$), or it can yield another constraint, which may be independent of $\phi_m$'s or not. If an independent constraint is obtained, one should require it to be preserved in time as well. This can yield yet another independent constraints or farther restrictions on $u$'s. The procedure ought to be continued until all the constraints are found. The additional constraints obtained in this way are called {\it secondary}. Following \cite{HT}, I shall denote the set of all the constraints by 
\be
\phi_j, \quad j=1,\dots,J ,
\ee
where $J-M$ is the number of secondary constraints. At the end, these constraints need to satisfy
\be\label{ueq}
\dot \phi_j\approx [\phi_j,H]_P+u^{m'}[\phi_j,\phi_{m'}]_{P}\approx 0 
\ee
(if they don't, more constraints or restrictions on $u$'s are needed). At any fixed point $(p,q)$ of the phase space, (\ref{ueq}) can be viewed as a system of linear equations for $u$'s. The general solution is provided by
\be\label{usol}
u^m=U^m+V^m , \quad V^m=v^a V^m_a ,
\ee
where $U$ is a particular solution of the inhomogeneous system, $V$ is the general solution of the homogeneous part, $V_a$ constitute an arbitrary basis for the space of solutions of the homogeneous equation and $v^a$ are completely arbitrary functions of time.  

A dynamical variable $F(q,p)$ is called {\it first class} if it commutes\footnote{Commutation means vanishing of the Poisson bracket in the case of bosonic variables and vanishing of the generalized Poisson bracket (discussed later) for fermionic ones.} with all the constraints. It is easy to verify that a modified Hamiltonian
\be
H'=H+U^m\phi_m ,
\ee
called the {\it first class Hamiltonian}, is first class. Note that $U$ is a particular solution of (\ref{ueq}) which can be chosen in a definite form, hence there are no arbitrary functions in $H'$, although one can construct $H'$ in meny different ways by choosing different particular solutions $U$. The total Hamiltonian (\ref{H_T}) that determines the dynamics of the system can now be rewritten in the form from which some of the arbitrary functions have been eliminated
\be\label{H_T2}
H_T=H'+v^a\phi_a ,\quad \phi_a:=V_a^m\phi_m .
\ee
The arbitrary functions $v^a$ remain present in the final form of the dynamical equations and indicate the presence of gauge freedom in the system, as explained below.

All the constraints can be separated into first class and second class constraints \footnote{A phase space function is called {\it second class} if it fails to satisfy the first class condition.}, which I shall denote by $\gamma_a$ and $\chi_{\alpha}$ respectively.  There are many possible realizations of the system of constraints in a given theory, all of them defining the same submanifold in phase space. For example, one can add the constraints, multiply them by functions or insert them as arguments of functions that vanish at zero\footnote{If the constraints are required to satisfy some regularity conditions (see 1.1.2 of \cite{HT}), than the set of possible operations that can be performed on the constraints is reduced to addition and multiplication by functions, since then the theorem holds that states that any function that vanishes on the constraint surface is necessarily a linear combination of the constraints (theorem 1.1 of \cite{HT}).}. Performing that kind of operations can change the total, as well as relative number of first and second class constraints. Indeed, adding a second class constraint to the first class one will result in a second class constraint. In fact, all the constraints could be made second class in this way, but this is not what we wish to achieve. I will say that the constraints are {\it well separated} into first and second class ones if the number of second class constraints is made minimal.

\subsection{Gauge freedom}

In classical physics, the time evolution of a system is expected to be deterministic. It appears that imposing such an assumption on a system with first class primary constraints leads to the conclusion that there is no one to one correspondence between points on the constraint surface and physical states. To see this, consider an arbitrary dynamical variable $F(t)\equiv F(q(t),p(t))$ whose dependence on $t$ is analytic and whose value at an instant of time $t_0$ is well established. The value of $F$ at time $t=t_0+\tau$ will be
\be
F(t_0+\tau)=F(t_0)+\tau \dot{F}(t_0)+\frac{\tau^2}{2}\ddot{F}(t_0)+\dots 
\ee
Using $\dot{F}=[F,H_T]$, $\ddot{F}=[[F,H_T],H_T]$ and (\ref{H_T2}) one obtains
\be\label{Ftrans}
\begin{aligned}
&F(t_0+\tau)=F+\tau\left({[F,H']+v^a[F,\phi_a]}\right) \\
&+\frac{\tau^2}{2}\left\{{[[F,H'],H']+v^a\left({2[[F,\phi_a],H']+[F,[H',\phi_a]]}\right)+v^av^{a'}[[F,\phi_a],\phi_{a'}]}\right\}+o(\tau^3)  ,
\end{aligned}
\ee
where on the RHS all the time dependent quantities (the functions $v^a$ and the brackets) should be evaluated at $t_0$. Now the functions $v^a$ can be prescribed arbitrarily. By adopting an alternative set of functions, say $\tilde{v}^a$, one gets different value $\tilde{F}(t_0+\tau)$ of the dynamical variable $F$ at $t_0+\tau$. Clearly, this difference cannot be physically meaningful if time evolution is to be deterministic. Rather, the difference should be interpreted as gauge freedom on which no measurable physical quantity should depend. Up to linear terms in $\tau$, this difference is simply
\be\label{ginf}
\delta F(t_0+\tau)=\tilde{F}(t_0+\tau)-F(t_0+\tau)=\tau\delta v^a(t_0) [F,\phi_a](t_0) , \quad \delta v^a:=\tilde{v}^a-v^a .
\ee
The infinitesimal form of the transformation, (\ref{ginf}), justifies the statement that the gauge transformation corresponding to the change in a particular function $v^a$ (for fixed $a$) is generated by the corresponding constraint $\phi_a$. The primary constraints $\phi_a$, as defined in (\ref{H_T2}), are first class, since for any constrain $\phi_j$ one has $[\phi_j,\phi_a]=[\phi_j,V_a^m]\phi_m+V_a^m[\phi_j,\phi_m]\approx 0$, where the vanishing of the last term follows straightforwardly form the definition of $V_a$'s as solutions to the homogeneous part of the system (\ref{ueq}). The conclusion follows that {\it first class primary constraints generate gauge transformations}. Consequently, if a dynamical variable $F$ is to represent an observable, it should commute with $\phi_a$'s. This guarantees that $F$ does not change under gauge transformations up to terms linear in $\tau$ (see (\ref{ginf})). However, in order to guarantee the invariance of $F$ up to second order terms in $\tau$, it is necessary to assume additionally that $F$ commutes with $[\phi_a,H']$ (look at the terms proportional to $\tau^2$ in (\ref{Ftrans})). Although $[\phi_a,H']$ is certainly a first class constraint \footnote{It is first class because the bracket of first class functions is of first class. This follows straightforwardly from the Jacobi identity. The bracket $[\phi_a,H']$ is also a constraint, since $H'$ is first class and hence its bracket with any constraint vanishes weakly.}, it does not have to be spanned by primary constraints. It follows that $F$ may in general need to commute with some secondary first class constraints, in order to be gauge invariant. If a dynamical variable commutes with all the first class constraints, both primary and secondary, than it is called a {\it classical observable}. Such observables are gauge invariant up to any order in the expansion (\ref{Ftrans}) and thus can describe physically measurable quantities.

Evolving the basic dynamical variables $q$ and $p$ from an instant $t_0$ to some instant $t$ with different choices of the arbitrary functions $v^a$ will yield a collection of points of phase space $(q(t),p(t))$, all of them describing the same physical state of the system at the instant $t$. The points in this collection are related by gauge transformations generated by the first class primary constraints. However, there may exist other points that also describe the same state, if there are secondary first class constraints present. In the cannonical approach to the quantization of constrained systems it is common to assume that all the first class constraints, both primary and secondary, generate gauge transformations, although this assumption may fail to be true for some special systems (see the counterexample to the Dirac conjecture in \cite{HT}). To make all the gauge freedom manifest in the equations of evolution, the so called {\it extended Hamiltonian} is adopted as a generator of the dynamics
\be
H_E=H'+w^b\gamma_b ,
\ee
where $b$ numbers all the first class constraints and $w^b$ are arbitrary functions (note that $H_E$ differs from $H_T$ by the presence of secondary first class constraints). It is the dynamics generated by $H_T$, and not $H_E$, that is equivalent to the Euler--Lagrange equations of (\ref{SL}). However, it is not difficult to see that the dynamics of any classical observable does not depend on the choice between $H_T$ and $H_E$. The general comparison of the evolution equations generated by $H_T$ and $H_E$ can be found in \cite{HT}. I will just illustrate the difference on the example of electrodynamics in the second part of this article. In fact, the reader does not have to bother by the discrepancy between $H_T$ and $H_E$ and the philosophy underlying the preferential use of $H_E$ over $H_T$ in the canonical analysis. In the gauge fixing approach to the quantization, which will be ultimately adopted in the second part of this paper, it does not matter which dynamics is utilized.

From the viewpoint of quantum theory, what we wish to extract from classical Hamiltonian analysis are the commutators between the physically important dynamical variables. The appropriate classical bracket acting on pairs of phase space functions should be identified that will then be replaced by (anti)commutators. In the absence of constraints and fermionic degrees of freedom, the Poisson bracket does the job. But it is not consistent neither with the constraints, nor with the presence of fermions (an antisymmetric structure cannot be consistently replaced by the symmetric anti--commutator). I shall now define the Dirac bracket, which is the modification of PB needed to handle the constraints in the absence of fermions. Later, the concept of the generalized Dirac bracket will be introduced that allows for inclusion of fermionic degrees of freedom.

\subsection{The Dirac bracket}

Assume that the set of all the constraints is well separated into the first class constraints $\gamma_b$ and the second class constraints $\chi_{\beta}$. The {\it Dirac bracket} (DB) of the dynamical variables $F$ and $G$ is defined by
\be\label{DB}
[F,G]_D=[F,G]_P-[F,\chi_{\beta}]_PC^{\beta\beta'}[\chi_{\beta'},G]_P ,
\ee
where $C^{\beta\beta'}$ represents the inverse matrix to $C_{\beta\beta'}:=[\chi_{\beta},\chi_{\beta'}]_P$, i.e. $C^{\beta\beta'}C_{\beta'\beta''}=\delta^{\beta}_{\beta''}$  (the matrix $C_{\beta\beta'}$ is necessarily invertible, since otherwise one can construct a first class constraint from $\chi_{\beta}$'s, which means that the constraints were not well separated. See Theorem 1.3 of \cite{HT}). It is easy to check that the Dirac bracket is anti--symmetric, obeys the Jacobi identity and Leibniz rule. Hence, the Dirac brackets can be replaced by commutators when passing to quantum theory in a consistent way. What is an advantage of  DB over PB is that DB of any dynamical variable with a second class constraint vanishes.
This allows for interpreting second class constraints as strong operator equations if a theory is quantized by the replacement of the Dirac brackets by the commutators of operators (the commutator of any operator with zero needs to vanish). 
Also, the DB of any variable with a first class variable is equal to the PB and therefore the equation of motion (\ref{dotFH_T}) can be rewritten as
\be\label{dotFDB}
\dot{F}\approx[F,H'+v^a\phi_a]_{DB} .
\ee

\vskip 0.2 in
\subsection{The anti-commuting Grassman variables and the generalized Dirac bracket}  

In the quantum theory of fields, the field operators that describe particles of integer value of spin can be characterized by appropriate commutation relations, whereas those of half integer spins obey anti--commutation relations. This distinction follows basically from the postulate of the Poincar\'e invariance, the assumption that the fields ought to be expressible as weighted integrals of annihilation and creation operators, and the requirement of causality.  See \cite{Wein} for more detailed justification of the connection between spin and statistics. 

As mentioned in the previous subsection, it is not possible to replace consistently the Dirac brackets of basic canonical variables by anti--commutators, since the two structures have incompatible symmetry. What is needed is the classical bracket that is antisymmetric if at least one of the variables is bosonic and symmetric if both the variables are fermionic. In order to obtain this structure it is necessary to introduce the two kinds of variables already at the classical level. I shall assume that the classical variables can be {\it even} (bosonic) or {\it odd} (fermionic), or they can be a sum of those. When these {\it Grassman variables} are introduced, their multiplication is no longer commutative. An even variable commutes with all the others, whereas odd variables anti--commute with one another. Let $q^i$, $\theta^{\alpha}$ constitute the set of positions of even and odd type, respectively. The multiplication rules are
\be
q^i q^j-q^j q^i=0,\quad \theta^{\alpha} q^i-q^i\theta^{\alpha}=0,\quad \theta^{\alpha}\theta^{\beta}+\theta^{\beta}\theta^{\alpha}=0 .
\ee
The only functions of the positions which I shall consider will be analytic in the odd variables
\be\label{SF}
f(q,\theta)=f_0(q)+f_{\alpha}(q)\theta^{\alpha}+f_{\alpha\beta}(q)\theta^{\alpha}\theta^{\beta}+\dots
\ee
Any such function is a sum of an even part $f_E$ and an odd part $f_O$:
\be
\begin{aligned}
&f(q,\theta)=f_E(q,\theta)+f_O(q,\theta) , \\
&f_E(q,\theta)=f_0(q)+f_{\alpha\beta}(q)\theta^{\alpha}\theta^{\beta}+\dots \\
&f_O(q,\theta)=f_{\alpha}(q)\theta^{\alpha}+f_{\alpha\beta\gamma}(q)\theta^{\alpha}\theta^{\beta}\theta^{\gamma}+\dots
\end{aligned}
\ee 

\subsubsection{Differentiation with respect to odd variables}

Under the infinitesimal variation of odd variables, the function (\ref{SF}) changes as
\be
\delta f=\delta\theta^{\alpha}\frac{\partial^Lf}{\partial\theta^{\alpha}}=\frac{\partial^Rf}{\partial\theta^{\alpha}}\delta\theta^{\alpha} .
\ee 
The above equation should be considered as definition of {\it right derivatives} and {\it left derivatives} with respect to odd variables. From (\ref{SF}) it is clear that the derivative of an odd function is even and vice versa, and hence the relations follow:
\be
\frac{\partial^L f_E}{\partial\theta^{\alpha}}=-\frac{\partial^Rf_E}{\partial\theta^{\alpha}},\quad
\frac{\partial^L f_O}{\partial\theta^{\alpha}}=\frac{\partial^Rf_O}{\partial\theta^{\alpha}} .
\ee
Let us consider an example of a function represented by a series that cuts off at the third order term in odd variables:
\be\label{SFexamp}
f(q,\theta)=f_0(q)+f_{\alpha}(q)\theta^{\alpha}+f_{\alpha\beta}(q)\theta^{\alpha}\theta^{\beta} .
\ee
The derivatives are then given by
\be
\frac{\partial^Lf}{\partial\theta^{\lambda}}=
f_{\lambda}-2f_{\alpha\lambda}\theta^{\alpha}+3f_{\alpha\beta\lambda}\theta^{\alpha}\theta^{\beta},\quad
\frac{\partial^Rf}{\partial\theta^{\lambda}}=
f_{\lambda}+2f_{\alpha\lambda}\theta^{\alpha}+3f_{\alpha\beta\lambda}\theta^{\alpha}\theta^{\beta} .
\ee
Variating this expressions with respect to $\theta$ yields
\be
\begin{aligned}
&\delta \frac{\partial^Lf}{\partial\theta^{\lambda}}=
\delta\theta^{\kappa}\left({-2f_{\kappa\lambda}+6f_{\kappa\beta\lambda}\theta^{\beta}}\right)=
\left({-2f_{\kappa\lambda}-6f_{\kappa\beta\lambda}\theta^{\beta}}\right)\delta^{\theta^\kappa} ,\\
&\delta \frac{\partial^Rf}{\partial\theta^{\lambda}}=
\delta\theta^{\kappa}\left({2f_{\kappa\lambda}-6f_{\kappa\beta\lambda}\theta^{\beta}}\right)=
\left({2f_{\kappa\lambda}+6f_{\kappa\beta\lambda}\theta^{\beta}}\right)\delta^{\theta\kappa} ,
\end{aligned}
\ee
from which the second derivatives with respect to odd variables can be read out. It follows that
\be\label{dirsym}
\begin{aligned}
\frac{\partial^L}{\partial\theta^{\kappa}}\left({\frac{\partial^Rf}{\partial\theta^{\lambda}}}\right)=
\frac{\partial^R}{\partial\theta^{\lambda}}\left({\frac{\partial^Lf}{\partial\theta^{\kappa}}}\right),\quad
\frac{\partial^L}{\partial\theta^{\kappa}}\left({\frac{\partial^Lf}{\partial\theta^{\lambda}}}\right)=
-\frac{\partial^L}{\partial\theta^{\lambda}}\left({\frac{\partial^Lf}{\partial\theta^{\kappa}}}\right),\quad
\frac{\partial^R}{\partial\theta^{\kappa}}\left({\frac{\partial^Rf}{\partial\theta^{\lambda}}}\right)=
-\frac{\partial^R}{\partial\theta^{\lambda}}\left({\frac{\partial^Rf}{\partial\theta^{\kappa}}}\right).\quad
\end{aligned}
\ee
The last two equalities imply that
\be\label{dirvanish}
\frac{\partial^L}{\partial\theta^{\underline{\kappa}}}\left({\frac{\partial^Lf}{\partial\theta^{\underline{\kappa}}}}\right)=
\frac{\partial^R}{\partial\theta^{\underline{\kappa}}}\left({\frac{\partial^Rf}{\partial\theta^{\underline{\kappa}}}}\right)=0 .
\ee
Here the bar below $\kappa$ means that Einstein summation convention is not applied, i.e. $\underline{\kappa}$ is a fixed value. The identities (\ref{dirsym}) and (\ref{dirvanish}) can be proved by induction to hold for any function of the form (\ref{SF}) and are necessary to prove Jacobi identity, to be discussed below. It is also easy to see that the derivatives with respect to even variables commute with those with respect to odd ones.

\subsubsection{Hamiltonian formalism in the presence of odd variables}

The action of the theory containing both even and odd positions is of the form
\be
S=\int L(q,\dot q,\theta,\dot \theta)dt ,
\ee
where the Lagrangian is assumed to be even (the time derivatives of even/odd variables are obviously of the same type as the original variables). The equations of motion that follow from the stationary condition for this action are
\be\label{ELodd}
\frac{\partial L}{\partial q^i}=\frac{d}{dt}\left({\frac{\partial L}{\partial \dot q^i}}\right) ,\quad
\frac{\partial^{L/R} L}{\partial\theta^{\alpha}}=\frac{d}{dt}\left({\frac{\partial^{L/R} L}{\partial \dot \theta^{\alpha}}}\right) ,\quad
\ee
where $L/R$ means that either type of the derivative can be used, if it is the same for both sides of the equation. The canonical momenta are defined by
\be
p_i:=\frac{\partial L}{\partial\dot q^i},\quad \pi^L_{\alpha}:=\frac{\partial^L L}{\partial\dot\theta^{\alpha}}, \quad
 \pi^R_{\alpha}:=\frac{\partial^R L}{\partial\dot\theta^{\alpha}} .
\ee
Since the Lagrangian is even, it follows that $\pi^L_{\alpha}=-\pi^R_{\alpha}$. Finally, the canonical Hamiltonian is given by
\be
H=\dot q^ip_i+\dot\theta^{\alpha}\pi^L_{\alpha}-L=p_i\dot q^i+\pi^R_{\alpha}\dot\theta^{\alpha}-L .
\ee
Note that one can use either left or right momenta, but the order of factors in the terms 
$\dot\theta^{\alpha}\pi^L_{\alpha}=-\pi^L_{\alpha}\dot\theta^{\alpha}$ and 
$\pi^R_{\alpha}\dot\theta^{\alpha}=-\dot\theta^{\alpha}\pi^R_{\alpha}$ has to be chosen correctly in order to avoid minus signs. 

It is also allowable to use left momenta for some odd variables and right momenta for the others. Such a mixed choice will appear to be particularly convenient in the case of the Dirac field, where I shall choose to right--differentiate with respect to the components of the field $\psi$ and left--differentiate with respect to the components of its Dirac conjugate $\ov{\psi}$. Before turning to the Dirac field in the next section, I shall now consider a theory with finite number of degrees of freedom, defined by the Lagrangian $L\left({q,\dot q,\theta,\ov{\theta},\dot{\theta},\dot{\ov{\theta}}}\right)$. 
I do not assume anything about the relation between the positions $\theta$ and $\ov{\theta}$, but I choose to right--differentiate with respect to $\theta$'s and left--differentiate with respect to $\ov{\theta}$'s. The momenta and the Hamiltonian are given by
\be\label{oddH}
\begin{aligned}
&p_i=\frac{\partial L}{\partial \dot q^i},\quad
\pi_{\alpha}:=\pi^R_{\alpha}=\frac{\partial^R L}{\partial\dot\theta^{\alpha}} ,\quad
\ov{\pi}_{\alpha}:=\ov{\pi}^L_{\alpha}=\frac{\partial^L L}{\partial\dot{\ov{\theta}}^{\alpha}},\\
&H=p_i\dot q^i+\pi_{\alpha}\dot\theta^{\alpha}+\dot{\ov{\theta}}^{\alpha}\ov{\pi}_{\alpha}-L .
\end{aligned}
\ee 
Variation of the Hamiltonian is given by
\be\label{dH1}
\delta H=\delta p_i\dot q^i+\delta\pi_{\alpha}\dot\theta^{\alpha}+\dot{\ov{\theta}}^{\alpha}\delta\ov{\pi}^{\alpha}-
\frac{\partial L}{\partial q^i}\delta q^i-
\frac{\partial^RL}{\partial\theta^{\alpha}}\delta\theta^{\alpha}-\delta\ov{\theta}^{\alpha}\frac{\partial^LL}{\partial\ov{\theta}^{\alpha}} ,
\ee
where the definitions of the momenta (\ref{oddH}) where used in the calculation. Also, the variation of $H$ interpreted as a function of positions and momenta is
\be\label{dH2}
\delta H=\frac{\partial H}{\partial q^i}\delta q^i+\frac{\partial H}{\partial p_i}\delta p_i+
\frac{\partial^RH}{\partial\theta^{\alpha}}\delta\theta^{\alpha}+\delta\ov{\theta}^{\alpha}\frac{\partial^LH}{\partial\ov{\theta}^{\alpha}}+
\delta\pi_{\alpha}\frac{\partial^LH}{\partial\pi_{\alpha}}+\frac{\partial^RH}{\partial\ov{\pi}_{\alpha}}\delta\ov{\pi}^{\alpha} .
\ee

\subsubsection{Unconstrained systems: the generalized Poisson bracket}

If the equations (\ref{oddH}) can be used to express velocities in terms of positions and momenta in a unique way and do not lead to any relations between positions and momenta, then the positions and momenta are independent and hence their general variations are linearly independent. It then follows from (\ref{dH1}) and (\ref{dH2}) that
\be\label{pomx}
\dot q^i=\frac{\partial H}{\partial p_i},\quad\dot\theta^{\alpha}=\frac{\partial H}{\partial\pi_{\alpha}},\quad \dot{\ov{\theta}}^{\alpha}=\frac{\partial H}{\partial\ov{\pi}_{\alpha}},\quad
\frac{\partial L}{\partial q^i}=-\frac{\partial H}{\partial q^i},\quad \frac{\partial L}{\partial \theta^{\alpha}}=-\frac{\partial H}{\partial\theta^{\alpha}},\quad \frac{\partial L}{\partial\ov{\theta}^{\alpha}}=-\frac{\partial H}{\partial \ov{\theta}^{\alpha}} .
\ee
Here and below the derivatives with respect to $\theta$'s and $\ov{\pi}$'s are defined as right derivatives, whereas those with respect to $\ov{\theta}$'s and $\pi$'s are left derivatives. 
In order to rewrite these equations in a form that does not include the velocities, the relations following from the Euler--Lagrange equations (\ref{ELodd})
\be
\frac{\partial L}{\partial q^i}=\dot{p}_i,\quad 
\frac{\partial L}{\partial \theta^{\alpha}}=\dot{\pi}_{\alpha},\quad 
\frac{\partial L}{\partial \ov{\theta}^{\alpha}}=\dot{\ov{\pi}}_{\alpha} 
\ee
need to be used. When inserted into (\ref{pomx}), they lead to
\be
\dot q^i=\frac{\partial H}{\partial p_i},\quad\dot\theta^{\alpha}=\frac{\partial H}{\partial\pi_{\alpha}},\quad \dot{\ov{\theta}}^{\alpha}=\frac{\partial H}{\partial\ov{\pi}_{\alpha}},\quad
\dot p_i=-\frac{\partial H}{\partial q^i},\quad \dot{\pi}_{\alpha}=-\frac{\partial H}{\partial\theta^{\alpha}},\quad \dot{\ov{\pi}}_{\alpha}=-\frac{\partial H}{\partial \ov{\theta}^{\alpha}} .
\ee
These Hamilton equations are equivalent to the Euler--Lagrange equations (\ref{ELodd}). The time derivative of any dynamical variable $F(q,p,\theta,\ov\theta,\pi,\ov\pi)$ is given by
\be\label{Fdot}
\begin{aligned}
&\dot F=\frac{\partial F}{\partial q^i}\dot q^i+\frac{\partial F}{\partial p_i}\dot p_i+\frac{\partial F}{\partial \theta^{\alpha}}\dot\theta^{\alpha}+\dot{\ov{\theta}}^{\alpha}\frac{\partial F}{\partial \ov{\theta}^{\alpha}}+
\dot\pi_{\alpha}\frac{\partial F}{\partial \pi_{\alpha}}+\frac{\partial F}{\partial\ov{\pi}_{\alpha}}\dot{\ov{\pi}}_{\alpha}\\
&=\frac{\partial F}{\partial q^i}\frac{\partial H}{\partial p_i}-
\frac{\partial F}{\partial p_i}\frac{\partial H}{\partial q^i}+\frac{\partial F}{\partial \theta^{\alpha}}\frac{\partial H}{\partial \pi_{\alpha}}+\frac{\partial H}{\partial \ov{\pi}_{\alpha}}\frac{\partial F}{\partial \ov{\theta}^{\alpha}}-
\frac{\partial H}{\partial \theta^{\alpha}}\frac{\partial F}{\partial \pi_{\alpha}}-\frac{\partial F}{\partial\ov{\pi}_{\alpha}}\frac{\partial H}{\partial\ov{\theta}^{\alpha}} .
\end{aligned}
\ee

In quantum theory, the dynamical variables will be replaced by operators. In order to accomplish the canonical quantization, we need to find the classical bracket $[,]_{GP}$, which I shall call the {\it generalized Poisson bracket} (GPB), for which the relation 
\be\label{corr}
[\widehat F,\widehat G]_{\mp}=i\widehat{[F,G]}_{GP}
\ee
will hold (the units in which $c=\hbar=1$ will be used throughout), where $F$ and $G$ are the classical dynamical variables and $\widehat F$, $\widehat G$ the corresponding operators. The operators that arise from even classical variables via the map  $ \ \widehat{} \ $ will be called {\it bosonic}, whereas those arising from odd variables {\it fermionic}. The bracket $[F,G]_{\mp}$ should be interpreted as commutator $[,]_-$ if at least one of the operators is bosonic. If both the operators are fermionic then $[,]_{\mp}$ denotes anti--commutator $[,]_+$. If the opearators do not have definite parity then they can be expressed as sums of those with well established parity and then the linearity of $[,]_{\mp}$ should be used\footnote{The map $ \ \widehat{} \ $ is assumed to be linear and all the brackets are bilinear with respect to $\mathbb{C}$ numbers.}. I wish that the time derivative (\ref{Fdot}) of any dynamical variable expresses through the classical bracket as $\dot F=[F,H]_{GP}$, which implies on account of (\ref{Fdot}) that
\be\label{FE}
[F,E]_{GP}
=\frac{\partial F}{\partial q^i}\frac{\partial E}{\partial p_i}-
\frac{\partial F}{\partial p_i}\frac{\partial E}{\partial q^i}+\frac{\partial F}{\partial \theta^{\alpha}}\frac{\partial E}{\partial \pi_{\alpha}}+\frac{\partial E}{\partial \ov{\pi}_{\alpha}}\frac{\partial F}{\partial \ov{\theta}^{\alpha}}-
\frac{\partial E}{\partial \theta^{\alpha}}\frac{\partial F}{\partial \pi_{\alpha}}-\frac{\partial F}{\partial\ov{\pi}_{\alpha}}\frac{\partial E}{\partial\ov{\theta}^{\alpha}} ,
\ee
where $E$ is any even dynamical variable (not necessarily the Hamiltonian) and the variable $F$ is completely arbitrary.
Now (\ref{corr}) implies that $[,]_{GP}$ needs to have the same algebraic properties as $[,]_{\mp}$. This means in particular that it is anti--symmetric whenever at least one of the variables is even and hence
\be\label{EF}
[E,F]_{GP}=-[F,E]_{GP} .
\ee
The formulas (\ref{FE}) and (\ref{EF}) specify GPB in the case when one of the variables is even. It only remained to find GPB in the case when both the variables are odd. If $A$, $B$ and $C$ are odd variables, then (\ref{corr}) implies that GPB should satisfy two algebraic conditions
\be\label{alsym}
[A,B]_{GP}=[B,A]_{GP},\quad
[A,BC]_{GP}=[A,B]_{GP}C-B[A,C]_{GP} .
\ee
The latter one follows from the operator identity $[\widehat A,\widehat B\widehat C]_-=[\widehat A,\widehat B]_+\widehat C-\widehat B[\widehat A,\widehat C]_+$ and from the fact that the product of two odd operators is even. I make an assumption that the classical bracket is composed from the products of partial derivatives with respect to basic canonical variables. More precisely, I assume that
\be\label{GPini}
[A,B]_{GP}=c\frac{\partial A}{\partial q^i}\frac{\partial B}{\partial p_i}+d\frac{\partial A}{\partial p_i}\frac{\partial B}{\partial q^i}+a\frac{\partial A}{\partial\theta^{\alpha}}\frac{\partial B}{\partial\pi_{\alpha}}+f\frac{\partial B}{\partial\theta^{\alpha}}\frac{\partial A}{\partial\pi_{\alpha}}+
b\frac{\partial B}{\partial\ov{\pi}_{\alpha}}\frac{\partial A}{\partial\ov{\theta}^{\alpha}}+g\frac{\partial A}{\partial\ov{\pi}_{\alpha}}\frac{\partial B}{\partial\ov{\theta}^{\alpha}}
\ee 
for some complex numbers $c$, $d$, $a$, $f$, $b$, $g$, to be determined from (\ref{alsym}). In the calculations that will follow it is important to remember that differentiation with respect to even variables does not change the type of parity, whereas differentiation with respect to odd variables reverses the parity, so e.g. $\partial A/\partial q^i$ is odd and $\partial A/\partial\theta^{\alpha}$ is even. Remembering this, it is straightforward to verify that
\be
\begin{aligned}
&[A,B]_{GP}-[B,A]_{GP}=(c+d)\left({\frac{\partial A}{\partial q^i}\frac{\partial B}{\partial p_i}+\frac{\partial A}{\partial p_i}\frac{\partial B}{\partial q^i}}\right)\\
&+(a-f)\left({\frac{\partial A}{\partial\theta^{\alpha}}\frac{\partial B}{\partial \pi_{\alpha}}-
\frac{\partial B}{\partial\theta^{\alpha}}\frac{\partial A}{\partial \pi_{\alpha}}}\right)+
(b-g)\left({\frac{\partial A}{\partial\ov{\pi}_{\alpha}}\frac{\partial B}{\partial \ov{\theta}^{\alpha}}-
\frac{\partial B}{\partial\ov{\pi}_{\alpha}}\frac{\partial A}{\partial\ov{\theta}^{\alpha}}}\right) .
\end{aligned}
\ee
If this expression is to vanish for {\it any} odd variables then it has to be $d=-c$, $f=a$, $g=b$. It follows that
\be
[A,B]_{GP}=c\left({\frac{\partial A}{\partial a^i}\frac{\partial B}{\partial p_i}-\frac{\partial A}{\partial p_i}\frac{\partial B}{\partial q^i}}\right)+
a\left({\frac{\partial A}{\partial\theta^{\alpha}}\frac{\partial B}{\partial\pi_{\alpha}}+\frac{\partial B}{\partial\theta^{\alpha}}\frac{\partial A}{\partial\pi_{\alpha}}}\right)+
b\left({\frac{\partial B}{\partial\ov{\pi}_{\alpha}}+\frac{\partial A}{\partial\ov{\pi}_{\alpha}}\frac{\partial B}{\partial\ov{\theta}^{\alpha}}}\right) .
\ee 
The remaining freedom of the parameters $c$, $a$, $b$ is eliminated when the second algebraic condition in (\ref{alsym}) is imposed. To calculate $[A,BC]_{GP}$ one should use (\ref{FE}) for $F=A$ and $E=BC$ (note that $BC$ is even as a product of odd variables and hance (\ref{FE}) applies). The derivatives of $BC$ with respect to even variables can be decomposed according to the standard Leibniz formula, but some care is necessary when differentiating with respect to odd variables. Specifically, the following relations hold
\be
\begin{aligned}
&\frac{\partial (BC)}{\partial\theta^{\alpha}}=-\frac{\partial B}{\partial\theta^{\alpha}}C+B\frac{\partial C}{\partial\theta^{\alpha}},\quad
\frac{\partial (BC)}{\partial\ov\theta^{\alpha}}=\frac{\partial B}{\partial\ov\theta^{\alpha}}C-B\frac{\partial C}{\partial\ov\theta^{\alpha}},\quad \\
&\frac{\partial (BC)}{\partial\ov\pi_{\alpha}}=-\frac{\partial B}{\partial\ov\pi_{\alpha}}C+B\frac{\partial C}{\partial\ov\pi_{\alpha}},\quad
\frac{\partial (BC)}{\partial\pi_{\alpha}}=\frac{\partial B}{\partial\pi_{\alpha}}C-B\frac{\partial C}{\partial\pi_{\alpha}} .
\end{aligned}
\ee
To prove these relations it is sufficient to calculate the variation of $BC$ and remember that all the derivatives with respect to $\theta^{\alpha}$ and $\ov\pi_{\alpha}$ are understood to be right derivatives and those with respect to $\ov\theta^{\alpha}$ and $\pi_{\alpha}$ are left. For example, a simple calculation of variation with respect to $\theta^{\alpha}$
\be
\delta (BC)=\delta BC+B\delta C=\frac{\partial^R B}{\partial\theta^{\alpha}}\delta\theta^{\alpha}C+
B\frac{\partial^R C}{\partial\theta^{\alpha}}\delta\theta^{\alpha}=
\left({-\frac{\partial B}{\partial\theta^{\alpha}}C+B\frac{\partial C}{\partial\theta^{\alpha}}}\right)\delta\theta^{\alpha}
\ee
proves the first of the identities. Having these results at hand it is straightforward to establish that the second identity of (\ref{alsym}) will be fulfilled if and only if $a=1$, $b=-1$ and $c=1$. The final form of the GPB for two odd variables can then be given:
\be\label{AB}
[A,B]_{GP}=\frac{\partial A}{\partial q^i}\frac{\partial B}{\partial p_i}-\frac{\partial A}{\partial p_i}\frac{\partial B}{\partial q^i}+
\frac{\partial A}{\partial\theta^{\alpha}}\frac{\partial B}{\partial\pi_{\alpha}}+\frac{\partial B}{\partial\theta^{\alpha}}\frac{\partial A}{\partial\pi_{\alpha}}-
\frac{\partial B}{\partial\ov{\pi}_{\alpha}}\frac{\partial A}{\partial\ov{\theta}^{\alpha}}-\frac{\partial A}{\partial\ov{\pi}_{\alpha}}\frac{\partial B}{\partial\ov{\theta}^{\alpha}} 
\ee 
which, together with (\ref{FE}) and (\ref{EF}) and the assumption of bilinearity of GPB, defines GPB for all the variables. This GPB is the same as the bracket given by the formula (4.1) of \cite{JHP}, although it is not so easily visible, since the authors of \cite{JHP} use left derivatives, whereas I use left derivatives when differentiating w.r.t. $\ov{\theta}$'s and $\pi$'s and right ones w.r.t. $\theta$'s and $\ov{\pi}$'s. 

The conditions (\ref{alsym}) where used to fix the parameters in the initial form of GPB (\ref{GPini}), 
so these conditions are certainly satisfied by the bracket from the construction. However, it is necessary to verify that the remaining algebraic conditions are satisfied by GPB. All these conditions are
\be\label{alsymall}
\begin{aligned}
&[F_1,E_1]_{GP}=-[E_1,F_1]_{GP},\\
&[A_1,A_2]_{GP}=[A_2,A_1]_{GP},\\
&[E_1,F_1F_2]_{GP}=[E_1,F_1]_{GP}F_2+F_1[E_1,F_2]_{GP},\\
&[A_1,E_1F_2]_{GP}=[A_1,E_1]_{GP}F_2+E_1[A_1,F_2]_{GP},\\
&[A_1,A_2F_1]_{GP}=[A_1,A_2]_{GP}F_1-A_2[A_1,F_1]_{GP},\\
&\left[{F_1,\left[{E_2,E_3}\right]_{GP}}\right]_{GP}+\left[{E_3,\left[{F_1,E_2}\right]_{GP}}\right]_{GP}+\left[{E_2,\left[{E_3,F_1}\right]_{GP}}\right]_{GP}=0,\\
&\left[{E_1,\left[{A_2,A_3}\right]_{GP}}\right]_{GP}-\left[{A_3,\left[{E_1,A_2}\right]_{GP}}\right]_{GP}+\left[{A_2,\left[{A_3,E_1}\right]_{GP}}\right]_{GP}=0,\\
&\left[{A_1,\left[{A_2,A_3}\right]_{GP}}\right]_{GP}+\left[{A_3,\left[{A_1,A_2}\right]_{GP}}\right]_{GP}+\left[{A_2,\left[{A_3,A_1}\right]_{GP}}\right]_{GP}=0 .
\end{aligned}
\ee
where $E_j$'s are even, $A_j$'s are odd and $F_j$'s are arbitrary. These identities can be derived straightforwardly from the corresponding operator identities under the assumption that the bracket corresponds to the anti--commutator if both variables are odd and the commutator if at least one of them is even. 
These conditions can be written more succinctly if the parity index $\#F$ is introduced, which is $0$ for even $F$ and $1$ for odd $F$. Then the conditions (\ref{alsymall}) can be rewritten in an equivalent form given by eqs. (4.3), (4.4), (4.5) and (4.6) of \cite{JHP}. However, when considering the examples, it is convenient to have them written down explicitly.

One can verify by straightforward but lengthy calculations (preferably performed with the help of {\it Mathematica} or {\it Maple}) that GPB defined by (\ref{FE}), (\ref{EF}) and (\ref{AB}) does indeed satisfy (\ref{alsymall}). To prove the Jacobi identities (the last three of (\ref{alsymall})), it is necessary to use (\ref{dirsym}) and (\ref{dirvanish}).

\subsubsection{Constrained systems: the generalized Dirac bracket}

In order to handle second class constraints consistently, it is necessary to introduce classical bracket which weakly vanishes whenever one of its arguments is a second class constraint. Can the formula (\ref{DB}) for the Dirac bracket be adopted in the presence of odd variables? Certainly, at lest one modification is necessary, namely all Poisson brackets have to be replaced by generalized ones. I shall define the {\it generalized Dirac bracket} (GDB) as
\be\label{GDB}
[F,G]_{GD}=[F,G]_{GP}-[F,\chi_{\beta}]_{GP}C^{\beta\beta'}[\chi_{\beta'},G]_{GP}.
\ee
This is the bracket that will be replaced by commutators and anti--commutators. Hence, in order for the quantization procedure to be consistent, (\ref{GDB}) ought to satisfy all the algebraic conditions (\ref{alsymall}) (just replace GP by GDB in (\ref{alsymall})). I shall assume for simplicity that all the second class constraints $\chi_{\beta}$ have the same Grassman parity (this assumption will be weakened somewhat below). Then $C^{\beta\beta'}$ is anti--symmetric in the even case and symmetric in the odd case. The matrix elements are even in both cases. Under this assumption it is straightforward to show that the identities (\ref{alsymall}) are satisfied indeed by (\ref{GDB}). Note however that the order of factors, as well as ordering of arguments of GP is important. For example, one could try to define GDB as
\be\label{GDB1}
[F,G]_{GD_1}=[F,G]_{GP}+[\chi_{\beta},F]_{GP}C^{\beta\beta'}[\chi_{\beta'},G]_{GP}
\ee
or
\be\label{GDB2}
[F,G]_{GD_2}=[F,G]_{GP}-[\chi_{\beta'},G]_{GP}C^{\beta\beta'}[F,\chi_{\beta}]_{GP},
\ee
both of these expressions being equivalent to (\ref{GDB}) so long as the constraints and all the variables are even. However, $GD_1$ will fail to be anti--symmetric in the case of $\chi_{\beta}$, $G$ odd and $F$ even. Indeed, one gets
\be
[E,A]_{GD_1}+[A,E]_{GD_1}=2[\chi_{\beta},E]_{GP}C^{\beta\beta'}[\chi_{\beta'},A]_{GP},
\ee
where $E$ is even and $A$ and $\chi_{\beta}$ odd.
 $GD_2$ will fail to satisfy the Leibniz rule for $\chi_{\beta}$ odd:
\be
[E,FA]_{GD_2}=[E,F]_{GD_2}A+F[E,A]_{GD_2}+2[\chi_{\beta'},F]_{GD_2}C^{\beta\beta'}[E,\chi_{\beta}]A
\ee
for $E$, $F$ even and $A$, $\chi_{\beta}$ odd. The reader is encouraged to verify that other alternatives for (\ref{GDB}) are not consistent with (\ref{alsymall}).

I assumed that all the second class constraints have the same Grassman parity. However, if these constraints can be separated into groups such that the constraints from different groups weakly commute with one another and all the constraints within a group have the same parity, then the results discussed above will also be true. This more general situation will occur in electrodynamics.

\section{Canonical quantization of the Dirac field}\label{dirac}

\subsection{Classical Hamiltonian analysis}

The action of the theory is
\be\label{Sdir}
S_D=\int \mathcal{L}_D\,d^4x ,\qquad
\mathcal{L}_D=\frac{i}{2}\left[{\overline{\psi}(x)\gamma^a\partial_a\psi(x)-\partial_a\overline{\psi}(x)\gamma^a\psi(x)}\right]-m\overline{\psi}(x)\psi(x) .
\ee
Here $(x)=(t,\vec{x})$ represents Minkowskian coordinates of flat space--time, $\psi$ can be thought of as a column of four complex valued functions  $\psi_l$, $l=1,\dots,4$ on space--time, $\gamma^a$ are the Dirac matrices, $\overline{\psi}:=\psi^{\dag}\gamma^0$ is the Dirac conjugation (here $\dag$ denotes the Hermitian conjugation of a column matrix), $a=0,\dots,3$ is a space--time index. For fixed $t$, $\psi_l(\vec{x})$ should be thought of as odd variables, so one needs to keep track of minus signs whenever the ordering of these fields is changed. However, consequent application of matrix notation in the calculations makes the ordering to be automatically correct in most cases and hence one can forget about the odd nature of $\psi$'s at the beginning of the analysis.

The Dirac Lagrangian density can be written in many equivalent ways, owing to the possibility of adding divergence of a vector field. The choice 
$\mathcal{L}_D$ differs by divergence of $V^a=\frac{i}{2}\overline{\psi}\gamma^a\psi$ from the simplest Lagrangian density
\be\label{Lsimp}
\mathcal{L}_{simp}=\overline{\psi}\left({i\gamma^a\partial_a-m}\right)\psi .
\ee
Although the two versions (as well as many others) are equivalent in flat space, they appear not to be equivalent when gravity, interpreted a Yang--Mills gauge theory of the Poincar\'e group, is minimally included \cite{Kazm1}. A modified coupling procedure can be introduced that is free of this ambiguity \cite{Kazm3}. This corrected coupling procedure is equivalent to the standard minimal one if $\mathcal{L}_D$, and not any other version of the Dirac Lagrangian, is used as a starting point for inclusion of gravity. Although this choice does not matter in the case of electrodynamics, I will chose $\mathcal{L}_D$, since then the extension of our considerations to the gravitational case will be possibly straightforward.

Calculation of the momenta in accordance with (\ref{momenta}) gives
\be
\pi:=\frac{\partial\mathcal{L}}{\partial(\partial_0\psi)}=\frac{i}{2}\overline{\psi}\gamma^0,\quad
\overline{\pi}:=\frac{\partial\mathcal{L}}{\partial(\partial_0\overline{\psi})}=-\frac{i}{2}\gamma^0\psi .
\ee
These relations does not involve time derivatives of fields and hence represent the primary constraints. In the matrix notation, the constraints can be written as
\be\label{chi}
\chi_1:=\pi-\frac{i}{2}\overline{\psi}\gamma^0,\quad \chi_2=\overline{\pi}+\frac{i}{2}\gamma^0\psi .
\ee
Hence, $\chi_1$ is a row matrix, whereas $\chi_2$ is a column. Strictly speaking, (\ref{chi}) represents infinite number of primary constraints, which can be labeled by $l$ and the space point $\vec{x}$
\be\label{chicomp}
\chi_{1l\vec{x}}=\pi_l(\vec{x})-\frac{i}{2}\overline{\psi}_{l'}(\vec{x})\gamma^0_{l'l},\quad
\chi_{2l\vec{x}}=\overline{\pi}_l(\vec{x})+\frac{i}{2}(\vec{x})\gamma^0_{ll'}\psi_{l'}.\quad
\ee

The basic canonical variables $\psi$, $\overline{\psi}$, $\pi$, $\overline{\pi}$ are all odd. The formulas (\ref{EF}), (\ref{FE}) and (\ref{AB}) for GPB reduce to
\be\label{GPdir}
[F,G]_{GP}=\int \left({
\frac{\delta F}{\delta\psi(\vec{x})} \frac{\delta G}{\delta\pi(\vec{x})} \pm 
\frac{\delta G}{\delta\overline{\pi}(\vec{x})} \frac{\delta F}{\delta\overline{\psi}(\vec{x})} \mp
\frac{\delta G}{\delta\psi(\vec{x})} \frac{\delta F}{\delta\pi(\vec{x})} -
\frac{\delta F}{\delta\overline{\pi}(\vec{x})} \frac{\delta G}{\delta\overline{\psi}(\vec{x})}
}\right) d^3x ,
\ee
where the upper sign applies whenever at least one of the variables $F$, $G$ is even and the lower one corresponds to $F$ and $G$ odd. The matrix index $l$ was omitted in favor of matrix multiplication between factors. If $F$ is a scalar functional of the canonical fields, then the functional derivatives ${\delta F}/{\delta \psi(\vec{x})}$ and ${\delta F}/{\delta \overline{\pi}(\vec{x})}$ are column matrices, whereas 
 ${\delta F}/{\delta \overline{\psi}(\vec{x})}$ and ${\delta F}/{\delta \pi(\vec{x})}$ are rows. To see this, consider an example of a functional
\be
F=\int\overline{\psi}(\vec{x})M(\vec{x})\psi(\vec{x}) d^3x ,
\ee
where $M$ is any matrix--valued function on space that does not depend on $\psi$ and $\overline{\psi}$. Under the infinitesimal change of $\psi$, the variation of $F$ is $\delta F=\int\overline{\psi}(\vec{x})M(\vec{x})\delta\psi(\vec{x}) d^3x$, whereas the variation of $F$ under the change of $\overline{\psi}$ is $\delta F=\int\delta\overline{\psi}(\vec{x})M(\vec{x})\psi(\vec{x}) d^3x$. It follows that 
\be
\frac{\delta F}{\delta\psi(\vec{x})}=\overline{\psi}(\vec{x})M(\vec{x}),\quad
\frac{\delta F}{\delta\overline{\psi}(\vec{x})}=M(\vec{x})\psi(\vec{x})
\ee
and hence  ${\delta F}/{\delta \psi(\vec{x})}$ is a row and  ${\delta F}/{\delta \overline{\psi}(\vec{x})}$ a column. Note that due to the matrix formalism the functional derivatives with respect to $\psi$ and $\overline{\pi}$ are automatically right derivatives and those with respect to $\overline{\psi}$ and $\pi$ are left (this is why I chose such convention when discussing general systems with odd variables).

The canonical Hamiltonian calculated in accordance with (\ref{momenta}) is 
\be
\begin{aligned}
&H=\int \mathcal{H}(\vec{x}) d^3x, \\
&\mathcal{H}=\pi\dot\psi+\dot{\overline{\psi}}-\mathcal{L}_D=\left({\pi-\frac{i}{2}\overline{\psi}\gamma^0}\right)\dot{\psi}+\dot{\overline{\psi}}\left({\overline{\pi}+\frac{i}{2}\gamma^0\psi}\right)
-i\overline{\psi}\gamma^j\partial_j\psi+m\overline{\psi}\psi+\frac{i}{2}\partial_j\left({\overline{\psi}\gamma^j\psi}\right) ,
\end{aligned}
\ee
where $j=1,2,3$ is a spatial index. Discarding the last boundary term, which would not contribute to the brackets of $H$ with other variables, and using the constraints, we get
\be
H=\int\overline{\psi}(\vec{x})\left({-i\gamma^j\partial_j+m}\right)\psi(\vec{x}) d^3x 
\ee
and 
\be\label{Hdir}
H_T=H+\int\left({\chi_1(\vec{x})u^1(\vec{x})+u^2(\vec{x})\chi_2(\vec{x})}\right)d^3x ,
\ee
where $u^1$ and $u^2$ are correspondingly a column and a row of complex--valued functions, as yet undetermined.

The consistency conditions (\ref{constrpres}) for time evolution of constraints can be solved in two ways. One can utilize matrix formalism and use (\ref{chi}),  (\ref{GPdir}) and (\ref{Hdir}) to obtain
\be\label{chiH_T}
\begin{aligned}
&[\chi_1(\vec{x}),H_T]_{GP}=-\frac{\delta H_T}{\delta\psi(\vec{x})}-\frac{i}{2}\frac{\delta H_T}{\delta\overline{\pi}(\vec{x})}\gamma^0=
-iu^2(\vec{x})\gamma^0-i\partial_j\overline{\psi}(\vec{x})\gamma^j-m\overline{\psi}\approx 0 , \\
&[\chi_2(\vec{x}),H_T]_{GP}=-\frac{\delta H_T}{\delta\overline{\psi}(\vec{x})}+\gamma^0\frac{i}{2}\frac{\delta H_T}{\delta\pi(\vec{x})}=
i\gamma^0 u^1(\vec{x})+i\gamma^j\partial_j\psi(\vec{x})-m\psi\approx 0 . 
\end{aligned}
\ee
Clearly, these equations can be solved by the appropriate choise of $u$'s. Therefore, no secondary constraints appear. The solution for $u$'s is given by
\be
U^1=-\gamma^0\gamma^j\partial_j\psi-im\gamma^0\psi,\quad
U^2=-\partial_j\overline{\psi}\gamma^j\gamma^0+im\overline{\psi}\gamma^0 .
\ee
These expressions provide a particular solution of an inhomogeneous system of equations discussed in (\ref{ueq}) and therefore uppercase letters are used (compare (\ref{usol})). The general solution to the homogeneous part of (\ref{ueq}) is equal to $0$ in this case.

Alternatively, instead of using matrices, one could rewrite the total Hamiltonian as
\be\label{H_Tcomp}
H_T=\int\left({-i\overline{\psi}_l(\vec{x})\gamma^j_{ll'}\partial_j\psi_{l'}(\vec{x})+m\overline{\psi}_l(\vec{x})+\chi_{1l\vec{x}}u^1_l(\vec{x})+u^2_l(\vec{x})\chi_{2l\vec{x}}}\right)d^3x
\ee
and use the component form of constraints (\ref{chicomp}). In the integrand of (\ref{H_Tcomp}) all the variables, together with $u^1_l$ and $u^2_l$, are odd and their chronology is therefore important. When calculating the brackets of $\chi_{1l\vec{x}}$ and $\chi_{2l\vec{x}}$ with the integrand of (\ref{H_Tcomp}), it is useful to use the commutation relations between the basic canonical variables
\be
\begin{aligned}
&[\psi_l(\vec{x}),\psi_{l'}(\vec{x}')]_{GP}=[\overline{\psi}_l(\vec{x}),\overline{\psi}_{l'}(\vec{x}')]_{GP}=
[\pi_l(\vec{x}),\pi_{l'}(\vec{x}')]_{GP}=[\overline{\pi}_l(\vec{x}),\overline{\pi}_{l'}(\vec{x}')]_{GP}=0, \\
&[\psi_l(\vec{x}),\pi_{l'}(\vec{x}')]_{GP}=[\pi_{l'}(\vec{x}'),\psi_{l}(\vec{x})]_{GP}=\delta_{ll'}\delta (\vec{x}-\vec{x}') ,\\
&[\overline{\psi}_l(\vec{x}),\overline{\pi}_{l'}(\vec{x}')]_{GP}=[\overline{\pi}_{l'}(\vec{x}'),\overline{\psi}_{l}(\vec{x})]_{GP}=-\delta_{ll'}\delta (\vec{x}-\vec{x}')
\end{aligned}
\ee
(note the minus sign in front of $\delta_{ll'}$ in the last expression) and the Leibniz rule for odd variables, $[A,BC]_{GP}=[A,B]_{GP}C-B[A,C]_{GP}$. This results can be derived from (\ref{GPdir}). Using them, one can calculate that
\be\label{C12}
\begin{aligned}
&C_{1l\vec{x},1l'\vec{x}'}:=[\chi_{1l\vec{x}},\chi_{1l'\vec{x}'}]_{GP}=0,\\
&C_{2l\vec{x},2l'\vec{x}'}:=[\chi_{2l\vec{x}},\chi_{2l'\vec{x}'}]_{GP}=0 ,\\
&C_{1l\vec{x},2l'\vec{x}'}:=[\chi_{1l\vec{x}},\chi_{2l'\vec{x}'}]_{GP}=i\gamma^0_{l'l}\delta (\vec{x}-\vec{x}')=[\chi_{2l'\vec{x}'},\chi_{1l\vec{x}}]_{GP}=:C_{2l'\vec{x}',1l\vec{x}} , \\
&[\chi_{1k\vec{y}},H_T(\vec{x})]_{GP}=\int\left({i\overline{\psi}_l(\vec{x})\gamma^j_{lk}\partial_j\delta(\vec{x}-\vec{y})-
\left({m\overline{\psi}_k(\vec{x})+i\gamma^0_{lk}u^2_l(\vec{x})}\right)\delta(\vec{x}-\vec{y})}\right)d^3x ,\\
&[\chi_{2k\vec{y}},H_T(\vec{x})]_{GP}=\int\left({i\gamma^j_{kl'}\partial_j\psi_{l'}(\vec{x})-m\psi_k(\vec{x})+i\gamma^0_{kl}u^1_l(\vec{x})}\right)\delta(\vec{x}-\vec{y})d^3x .
\end{aligned}
\ee
After the integrations are performed, the last two expressions reduce to (\ref{chiH_T}). When writing $\partial_j\delta (\vec{x}-\vec{y})$ I always mean the differentiation with respect to first variable, i.e. 
$\partial_j\delta(\vec{x}-\vec{y})\equiv(\partial_j\delta)(\vec{x}-\vec{y})=\partial/\partial x^j\delta(\vec{x}-\vec{y})=-\partial/\partial y^j\delta(\vec{y}-\vec{x})=-\partial_j\delta(\vec{y}-\vec{x})$. So, unlike $\delta$, the function $\partial_j\delta$ is odd in the sense that $\partial_j\delta(-\vec{x})=-\partial_j\delta(\vec{x})$. The only nontrivial brackets that have to be evaluated in the calculations above are the brackets of the fields with the derivatives of their canonical conjugates, such as
\be
[\pi_k(\vec{y}),\partial_j\psi_{l}(\vec{x})]_{GP}=\frac{\delta\partial_j\psi_{l}(\vec{x})}{\delta\psi_k(\vec{y})}=-\delta_{kl}\partial_j\delta(\vec{y}-\vec{x})=\delta_{kl}\partial_j\delta(\vec{x}-\vec{y}) .
\ee
The easiest way to obtain this result is to write the derivative as $\partial_j\psi_l(\vec{x})=\int\delta(\vec{y}-\vec{x})\delta_{lk}\partial_j\psi_k(\vec{y}) d^3y$ and then variate,
 $\delta\partial_j\psi_l(\vec{x})=\int\delta(\vec{y}-\vec{x})\delta_{lk}\partial/\partial y^j\delta\psi_k(\vec{y}) d^3y=
-\int \partial/\partial y^j\delta(\vec{y}-\vec{x})\delta_{lk}\delta\psi_k(\vec{y}) d^3y$, where in the last step the boundary term was omitted.

\vskip 0.2 in
\centerline{\bf Equations of motion}

Since the first class constraints are not present in the system, it follows that 
\be
\begin{aligned}
&H_T=H_E=H'=H+\int\left({\chi_1(\vec{x})U^1(\vec{x})+U^2(\vec{x})\chi_2(\vec{x})}\right)d^3x \\
&=\int\left[{\overline{\psi}\left({-i\gamma^j\partial_j+m}\right)\psi-
\left({\pi-\frac{i}{2}\overline{\psi}\gamma^0}\right)\gamma^0\left({\gamma^j\partial_j\psi+im\psi}\right)-
\left({\partial_j\overline{\psi}\gamma^j-im\overline{\psi}}\right)\gamma^0\left({\overline{\pi}+\frac{i}{2}\gamma^0\psi}\right)}\right]d^3x .
\end{aligned}
\ee
In the last form I dropped the argument $\vec{x}$ for simplicity. The equations of motion are
\be\label{EM}
\begin{aligned}
&\dot\psi=[\psi,H']_{GP}=-\gamma^0\left({\gamma^j\partial_j\psi+m\psi}\right),\\
&\dot{\overline{\psi}}=[\overline{\psi},H']_{GP}=-\left({\partial_j\overline{\psi}\gamma^j-im\overline{\psi}}\right)\gamma^0,\\
&\dot\pi=[\pi,H']_{GP}=-\partial_j\pi\gamma^0\gamma^j+im\pi\gamma^0,\\
&\dot{\overline{\pi}}=[\overline{\pi},H']_{GP}=-\gamma^j\gamma^0\partial_j\overline{\pi}-im\gamma^0\overline{\pi}.
\end{aligned}
\ee
The equalities of the form $[f,H']_{GP}=g$, where $f$ and $g$ are functions defined on space, should be understood as equalities of functions, i.e. $\forall \vec{x}, \ [f(\vec{x}),H']=g(\vec{x})$. 
The equations of motion (\ref{EM}) need to be supplemented with the constraints. It  is easy to see that all the equations then reduce to the Dirac equation
\be
\left({i\gamma^a\partial_a-m}\right)\psi=0 .
\ee

\subsection{The generalized Dirac bracket and the equal time anti--commutators of field operators}

Having the matrix $C_{1l\vec{x},2l'\vec{x}'}$ given by (\ref{C12}), one can seek for its inverse by imposing the conditions
\be
\begin{aligned}
&\sum_{k'=1}^4 \int C^{1k\vec{x},2k'\vec{x}'}C_{2k'\vec{x}',1l\vec{y}} \, d^3x'=
\sum_{k'=1}^4 \int C^{2k\vec{x},1k'\vec{x}'}C_{1k'\vec{x}',2l\vec{y}} \, d^3x'=
\delta_{kl}\delta(\vec{x}-\vec{y}) , \\
&\sum_{k'=1}^4 \int C^{1k\vec{x},1k'\vec{x}'}C_{1k'\vec{x}',2l\vec{y}} \, d^3x'=
\sum_{k'=1}^4 \int C^{2k\vec{x},2k'\vec{x}'}C_{2k'\vec{x}',1l\vec{y}} \, d^3x'=0 .
\end{aligned}
\ee
A unique solution is given by
\be
C^{1l\vec{x},2l'\vec{x}'}=C^{2l'\vec{x}',1l\vec{x}}=-i\gamma^0_{ll'}\delta(\vec{x}-\vec{x}') 
\ee
and the generalized Dirac bracket is
\be
\begin{aligned}
&[F,G]_{GD}=[F,G]_{GP}-\sum_{l,l'}\int\int
\left({[F,\chi_{1l\vec{x}}]_{GP}C^{1l\vec{x},2l'\vec{x}'}[\chi_{2l'\vec{x}'},G]_{GP}+
[F,\chi_{2l'\vec{x}'}]_{GP}C^{2l'\vec{x}',1l\vec{x}}[\chi_{1l\vec{x}},G]_{GP}}\right)d^3xd^3x' \\
&=[F,G]_{GP}+\sum_{l,l'} i\gamma^0_{ll'}\int
\left({
[F,\chi_{1l\vec{x}}]_{GP}[\chi_{2l'\vec{x}},G]_{GP}+
[F,\chi_{2l'\vec{x}}]_{GP}[\chi_{1l\vec{x}},G]_{GP}
}\right)d^3x
\end{aligned}
\ee
with $\chi_{1l\vec{x}}$ and $\chi_{2l\vec{x}}$ given by (\ref{chicomp}). The brackets of basic canonical variables can now be computed
\be\label{eqtime}
\begin{aligned}
&[\psi_l(\vec{x}),\psi_{l'}(\vec{x}')]_{GD}=[\pi_l(\vec{x}),\pi_{l'}(\vec{x}')]_{GD}=[\overline{\psi}_l(\vec{x}),\overline{\psi}_{l'}(\vec{x}')]_{GD}=[\overline{\pi}_l(\vec{x}),\overline{\pi}_{l'}(\vec{x}')]_{GD}=0,\\
&[\psi_l(\vec{x}),\pi_{l'}(\vec{x}')]_{GD}=[\pi_{l'}(\vec{x}'),\psi_l(\vec{x})]_{GD}=\frac{1}{2}\delta_{ll'}\delta(\vec{x}-\vec{x}'),\\
&[\overline{\psi}_l(\vec{x}),\overline{\pi}_{l'}(\vec{x}')]_{GD}=[\overline{\pi}_{l'}(\vec{x}'),\overline{\psi}_l(\vec{x})]_{GD}=-\frac{1}{2}\delta_{ll'}\delta(\vec{x}-\vec{x}'),\\
&[\psi_l(\vec{x}),\overline{\pi}_{l'}(\vec{x}')]_{GD}=[\overline{\psi}_l(\vec{x}),\pi_{l'}(\vec{x}')]_{GD}=0,\\
&[\psi_l(\vec{x}),\overline{\psi}_{l'}(\vec{x}')]_{GD}=[\overline{\psi}_{l'}(\vec{x}'),\psi_l(\vec{x})]_{GD}=-i\gamma^0_{ll'}\delta(\vec{x}-\vec{x}'),\\
&[\pi_l(\vec{x}),\overline{\pi}_{l'}(\vec{x}')]_{GD}=[\overline{\pi}_{l'}(\vec{x}'),\pi_l(\vec{x})]_{GD}=-\frac{i}{4}\gamma^0_{ll'}\delta(\vec{x}-\vec{x}').
\end{aligned}
\ee
The equal time anti--commutation relations of field operators can now be readily established according to the procedure
\be\label{quant}
[\widehat{F},\widehat{G}]_+=i\widehat{[F,G]}_{GD} ,
\ee 
which is the generalized version of (\ref{corr}) (note that $GP$ is now replaced by $GD$). The second class constraints (\ref{chi}) are interpreted as strong operator equations  in quantum theory and hence not all the basic field operators $\widehat{\psi}$, $\widehat{\overline{\psi}}$, $\widehat{\pi}$, $\widehat{\overline{\pi}}$ are independent. For example, one could chose to use $\widehat{\psi}$ and $\widehat{\overline{\psi}}$ as independent and interpret the remaining basic operators as derived from these in agreement with the constraints. Only the anti--commutation relations of $\psi$ and $\overline{\psi}$ are then necessary
\be\label{eqt}
\left[{\widehat{\psi}_l(\vec{x}),\widehat{\psi}_{l'}(\vec{x}')}\right]_+=0,\qquad
\left[{\widehat{\overline{\psi}}_l(\vec{x}),\widehat{\overline{\psi}}_{l'}(\vec{x}')}\right]_+=0,\qquad
\left[{\widehat{\psi}_l(\vec{x}),\widehat{\overline{\psi}}_{l'}(\vec{x}')}\right]_+=\gamma^0_{ll'}\delta(\vec{x}-\vec{x}') .
\ee
All the other anti--commutators can be derived from (\ref{eqt}). The reader is encouraged to verify that the results thus obtained agree with those that could be derived directly from  (\ref{eqtime}) through (\ref{quant}) for each bracket separately. This follows from the fact the the Dirac bracket is consistent with the second class constraints.

\subsection{How the anti--commutation relations should not be derived}

In many treatments of quantum field theory (e.g. \cite{Peskin}\cite{Wein}), the canonical analysis of the Dirac field begins with the introduction of the Hamiltonian formalism for field theory of systems that do not contain neither odd variables nor constraints (usually on the example of the Klein--Gordon field). Then the simplest Lagrangian density (\ref{Lsimp}) for the Dirac field is introduced,
for which one of the constraints
\be\label{wiaz}
\pi=i\overline{\psi}\gamma^0
\ee
is derived. The conclusion is then drawn that $\overline{\psi}$ is not an independent field, but rather a function of $\pi$. 

Although the references mentioned do not do this, it may be tempting, especially for a student who has not yet been introduced with the formalism of constraints or odd variables, to proceed with the derivation of the anti--commutation relations in the following way:
The usual PB for $\psi$ and $\overline{\psi}$
can be calculated, under the assumption that $\psi$ is expressed through $\pi$ according to (\ref{wiaz}). One could then try to obtain the anti--commutation relation for field operators via (\ref{quant}), with GDB replaced by PB. The result thus obtained coincides with (\ref{eqt}). Have we just circumvented the formalism of constraint systems and odd variables? Is then the concept of generalized Dirac bracket really necessary?

Although this ``derivation'' yields  correct anti--commutation relations (\ref{eqt}), it is important to stress that its correctness is conditioned by particular choice of Lagrangian density from the class of equivalence (i.e. $\mathcal{L}_{simp}$ (\ref{Lsimp}) and not, say, $\mathcal{L}_D$) (\ref{Sdir})), as well as particular ordering of arguments of the brackets (if the anti--commutation relation for $\overline{\psi}$ and $\psi$ was established in a way described above, then the result would obviously differ from the correct one by sign, due to the anti--symmetry of PB).

\subsection{Non--equal time commutation relations}

The non--equal time commutation relations for causal fields in flat Minkowski space are usually derived in a way that does not employ canonical formalism at all, but instead relies heavily on the Poincar\'e symmetry (see \cite{Wein} or \cite{Kazm5}). On the other hand, the derivation of (\ref{eqt}) presented here does not invoke the Poincar\'e symmetry at all. This is of importance, since the frequent usage of this symmetry in QFT is one of the most important obstacles for the straightforward extension of its apparatus  to the case of curved space--time, which may not posses any symmetries at all. One of the aims of this article is to show how far one can go with the quantization by using the canonical formalism and not the space--time symmetries.

To derive non--equal time anti--commutation relation for the Dirac field and its conjugate, consider the two space--time points $(x)=(t,\vec{x})$ and $(x')=(t+\tau,\vec{x}')$. Assuming the analyticity of $\psi$ in $t$, the GD for $\psi(x')$ and $\psi(x)$ can be represented by a series
\be\label{series}
\left[{\psi_{l'}(x'),\overline{\psi}_l(x)}\right]_{GD}=
\left[{\psi_{l'}(t+\tau,\vec{x}'),\overline{\psi}_l(t,\vec{x})}\right]_{GD}=
\sum_{n=0}^{\infty}\frac{\tau^n}{n!}
\left[{\psi^{(n)}_{l'}(t,\vec{x}'),\overline{\psi}_l(t,\vec{x})}\right]_{GD} ,
\ee
where $\psi^{(n)}$ is the $n$--th derivative of $\psi$ with respect to time. In order to shrink the series into something finite and simple, it is helpful to be organized.
I shall now prove the
\begin{theorem}
The following relation holds
\be
\psi^{(n)}(t,\vec{x})=\mathcal{B}^n\psi(t,\vec{x}) ,
\ee
where $\mathcal{B}^n$ denotes the $n$--th power of the matrix--differential operator
\be
\mathcal{B}:=-\gamma^0\gamma^j\partial_j-im\gamma^0 .
\ee
\end{theorem}
\begin{proof}
I shall use induction. For $n=0$ the result is trivially true and for $n=1$ it follows from the first equation of (\ref{EM}). I will show that for each $n\in\mathbb{N}$ the inductive assumption
\be\label{assum}
\psi^{(n)}(\vec{x})=\mathcal{B}^n\psi(\vec{x})
\ee
implies that
\be\label{thesis}
\psi^{(n+1)}(\vec{x})=\mathcal{B}^{n+1}\psi(\vec{x}) .
\ee
Note that
\be\label{niewiem}
\psi^{(n+1)}(\vec{x})=[\psi^{(n)}(\vec{x}),H']_{GP}=
\int\frac{\delta \mathcal{B}^n\psi(\vec{x})}{\delta\psi(\vec{z})}\frac{\delta H'}{\delta \pi(\vec{z})}d^3z=
\int\frac{\delta \mathcal{B}^n\psi(\vec{x})}{\delta\psi(\vec{z})}\mathcal{B}\psi(\vec{z})d^3z .
\ee
In the second step the inductive assumption (\ref{assum}) was used.
Since
\be
\mathcal{B}^n\psi(\vec{x})=\int \delta(\vec{z}-\vec{x})\mathcal{B}^n\psi(\vec{z}) d^3z=
\int \widetilde{\mathcal{B}^n}\delta(\vec{z}-\vec{x})\psi(\vec{z}) d^3z ,
\ee
where $\widetilde{X}$ denotes the differential operator obtained from $X$ by the change of signs in front of the terms that are of even rank in derivatives (e.g. $\widetilde{\mathcal{B}}=\gamma^0\gamma^j\partial_j-im\gamma^0$), it follows that
\be\label{deltaB}
\frac{\delta \mathcal{B}^n\psi(\vec{x})}{\delta\psi(\vec{z})}=\widetilde{\mathcal{B}^n}\delta(\vec{z}-\vec{x}).
\ee
Inserting this result to (\ref{niewiem}) leads to
\be
\psi^{(n+1)}(\vec{x})=\int\widetilde{\mathcal{B}^n}\delta(\vec{z}-\vec{x})\mathcal{B}\psi(\vec{z})d^3z=
\int\delta(\vec{z}-\vec{x})\mathcal{B}^{n+1}\psi(\vec{z})d^3z=\mathcal{B}^{n+1}\psi(\vec{x}) ,
\ee
which proves the thesis (\ref{thesis}).
\end{proof}
Let us now calculate the bracket that occurs under the summation sign in (\ref{series}). Using (\ref{GDB}) with $\vec{x}$ replaced by $\vec{z}$ and $l,l'$ by $k,k'$ in order to avoid conflicting indices, one gets
\be
\begin{aligned}
&\left[{\psi^{(n)}_{l'}(\vec{x}'),\overline{\psi}_l(\vec{x})}\right]_{GD}=
i\gamma^0_{kk'}\int 
\left[{\psi^{(n)}_{l'}(\vec{x}'),\chi_{1k\vec{z}}}\right]_{GP} 
\left[{\chi_{2k'\vec{z}},\overline{\psi}_l(\vec{x})}\right]_{GP}\\
&=-i\gamma^0_{kl}\left[{\psi^{(n)}_{l'}(\vec{x}'),\pi_k(\vec{x})}\right]_{GP}=
-i\gamma^0_{kl}\frac{\delta\psi^{(n)}_{l'}(\vec{x}')}{\delta\psi_k(\vec{x})}=
-i\gamma^0_{kl}\frac{\delta\left({\mathcal{B}^n\psi(\vec{x}')}\right)_{l'}}{\delta\psi_k(\vec{x})} \quad.
\end{aligned}
\ee
Using (\ref{deltaB}) and $\widetilde{\mathcal{B}^n}\delta(\vec{z}-\vec{x})=\mathcal{B}^n\delta(\vec{x}-\vec{z})$, which holds on account of the fact that
even derivatives of the Dirac delta are even functions and odd derivatives are odd functions, one finally gets
\be\label{psinbpsi}
\left[{\psi^{(n)}_{l'}(\vec{x}'),\overline{\psi}_l(\vec{x})}\right]_{GD}=-i\left({\mathcal{B}^n\delta(\vec{x}'-\vec{x})\gamma^0}\right)_{l'l} .
\ee
In order to shrink (\ref{series}) to the finite expression, I need to prove yet another
\begin{theorem}\label{th2}
The following identity holds
\be\label{theorem2}
-i\sum_{n=0}^{\infty}\frac{\tau^n}{n!}\mathcal{B}^n\delta(\vec{x})\gamma^0=
-i\left({i\gamma^a\partial_a+m}\right)\triangle (\tau,\vec{x}) ,
\ee
where
\be\label{triangle}
\begin{aligned}
&\triangle(t,\vec{x}):=
\int d\Gamma_p\left({e^{-iE_pt}e^{i\vec{p}\vec{x}}-e^{iE_pt}e^{-i\vec{p}\vec{x}}}\right) ,\\
&d\Gamma_p:=\frac{d^3p}{(2\pi)^32E_p},\qquad E_p:=\sqrt{{\vec{p}\,}^2+m^2}\quad.
\end{aligned}
\ee
\end{theorem}
\begin{proof}
It is important to stress that the three--vector $\vec{p}=(p^1,p^2,p^3)$ is just an integration variable, which could be called in any other way, i.e. no interpretation of $\vec{p}$ as particle's momentum is necessary.
To prove the theorem, I will first introduce the
\begin{lemma}
For any $k\in\mathbb{N}$ the following identity holds
\be\label{lemma}
(\partial_j\partial_j-m^2)^k \delta(\vec{x})=i\triangle^{(2k+1)}(0,\vec{x}) .
\ee
\end{lemma}
\begin{proof}
Straightforward calculation shows that
\be
(\partial_j\partial_j-m^2)\left({e^{i\vec{p}\vec{x}}+e^{-i\vec{p}\vec{x}}}\right)
=(iE_p)^2\left({e^{i\vec{p}\vec{x}}+e^{-i\vec{p}\vec{x}}}\right) ,
\ee
which easily generalizes to
\be\label{pom}
(\partial_j\partial_j-m^2)^k\left({e^{i\vec{p}\vec{x}}+e^{-i\vec{p}\vec{x}}}\right)
=(iE_p)^{2k}\left({e^{i\vec{p}\vec{x}}+e^{-i\vec{p}\vec{x}}}\right) .
\ee
Decomposing the Dirac delta as
\be
\delta(\vec{x})=\frac{1}{2(2\pi)^3}\int \left({e^{i\vec{p}\vec{x}}+e^{-i\vec{p}\vec{x}}}\right) d^3p ,
\ee
acting on it by $(\partial_j\partial_j-m^2)^k$ and using (\ref{pom}) yields
\be\label{pom1}
(\partial_j\partial_j-m^2)^k \delta(\vec{x})=
-i\int (iE_p)^{2k+1}\left({e^{i\vec{p}\vec{x}}+e^{-i\vec{p}\vec{x}}}\right) d\Gamma_p .
\ee
On the other hand, differentiation of (\ref{triangle}) gives
\be
\triangle^{(2k+1)}(t,\vec{x})=-\int (iE_P)^{2k+1} 
\left({e^{-iE_pt}e^{i\vec{p}\vec{x}}+e^{iE_pt}e^{-i\vec{p}\vec{x}}}\right) d\Gamma_p
\ee
which, for $t=0$, coincides with (\ref{pom1}), thus proving (\ref{lemma}).
\end{proof}
Having the Lemma proved, it is easy to prove the Theorem by summing (\ref{theorem2}) over even and odd values of $n$ separately. Since
\be
\mathcal{B}^2=-\gamma^i\gamma^j\partial_i\partial_j-m^2=\partial_j\partial_j-m^2
\ee
(use $[\gamma^a,\gamma^b]_+=2\eta^{ab}$, where $\eta=diag(1,-1,-1,-1)$), it follows that
\be
\mathcal{B}^{2k}\delta(\vec{x})=(\partial_j\partial_j-m^2)^k\delta(\vec{x})=i\triangle^{(2k+1)}(0,\vec{x}),\qquad
\mathcal{B}^{2k+1}\delta(\vec{x})=\mathcal{B}i\triangle^{(2k+1)}(0,\vec{x})
\ee
and the LHS of (\ref{theorem2}) is
\be
\begin{aligned}
&-i\sum_{n=0}^{\infty}\frac{\tau^n}{n!}\mathcal{B}^n\delta(\vec{x})\gamma^0=
\sum_{k=0}^{\infty}\frac{\tau^{2k+1}}{(2k+1)!}\mathcal{B}\triangle^{2k+1}(0,\vec{x})\gamma^0+
\sum_{k=0}^{\infty}\frac{\tau^{2k}}{(2k)!}\dot{\triangle}^{(2k)}(0,\vec{x})\gamma^0\\
&=\mathcal{B}\gamma^0\triangle(\tau,\vec{x})+\dot{\triangle}(\tau,\vec{x})\gamma^0=
-i\left({i\gamma^a\partial_a+m}\right)\triangle(\tau,\vec{x}).
\end{aligned}
\ee
\end{proof}
Insertion of (\ref{psinbpsi}) into (\ref{series}) and subsequent application of Theorem (\ref{th2}) gives non--equal time anti--commutation relation
\be\label{eq}
\left[{\widehat{\psi}_l(x),\widehat{\overline{\psi}}_{l'}(x')}\right]_+=
(i\gamma^a\partial_a+m)_{ll'}\triangle(x-x')
\ee
or, equivalently,
\be
\left[{\widehat{\psi}_l(x),{\widehat{\psi}}^{\dag}_{l'}(x')}\right]_+=
\left\{{(i\gamma^a\partial_a+m)\gamma^0}\right\}_{ll'}\triangle(x-x') .
\ee
It is important to stress that the formula (\ref{series}) already shows that the form of non--equal time anti--commutation relations is fully determined by the canonical Hamiltonian formalism. Neither the Lorentz symmetry of the Dirac field nor the symmetries of space--time are needed to obtain this relations. All the work done afterwards was aimed to simplify the formula for the anti--commutator of $\psi$ and $\overline{\psi}$ and here the symmetries were certainly helpful. However, even if that kind of simplification was not possible for technical reasons (e.g. in the generic curved space), the series formula similar to (\ref{series}) would still be obtainable and could be used in principle to finally compute measurable quantities from the theory.

\subsection{Quantization}

To quantize the theory, it is necessary to established the (anti)commutation relations between the basic field operators and all the other operators that represent physically important observables. Then it only remains to find a representation of these relations in a Hilbert space.
The anti--commutation relations between the basic canonical operators have been established. The most important observable is the energy represented by the Hamiltonian. But which Hamiltonian? Since there are no first class constraints in the system, only $H'$ and $H$ are at our disposal. But they differ by second class constraints which strongly vanish in the quantum theory. Hence, the operators corresponding to $H$ and $H'$ are equal. Recall from (\ref{dotFDB}) that the equations for time evolution (\ref{EM}) can be rewritten in terms of GDB, instead of GP, e.g.
\be
\dot\psi\approx[\psi,H]_{GDB}\approx-\gamma^0\left({\gamma^j\partial_j\psi+m\psi}\right)
\ee
(the prime at $H$ was omitted, since the GDB of any variable with a linear combination of second class constraints weakly vanishes). Then the application of (\ref{quant}) and setting second class constraints to zero yields the operator version of the Dirac equation
\be\label{DEop}
\left({i\gamma^a\partial_a-m}\right)\widehat{\psi}=0 .
\ee
The equation is sufficiently simple that a general solution can be given. This is very pragmatical, since then, if $\widehat{\psi}$ assumes a form of general solution to (\ref{DEop}), the commutation relations of $\psi$ and $\overline{\psi}$ with the Hamiltonian are automatically satisfied and it only remains to impose the equal time anti--commutation relations of fields (\ref{eqt}) (and also the commutation relations of fields with other observables such as momentum and electric charge, which I shall not discuss here). 

Acting by $i\gamma^b\partial_b+m$ on the LHS of (\ref{DEop}) yields the Klein--Gordon equation,
$(\square +m^2)\widehat{\psi}(x)=0$, from which it follows that $\widehat{\psi}$ is necessarily of the form
\be\label{psisol}
\widehat{\psi}(x)=\int\left({e^{-ip\cdot x}A(\vec{p})+e^{ip\cdot x}B(\vec{p})}\right) d\Gamma_p ,
\ee
where $p\cdot x=p^0 t-\vec{p}\cdot \vec{x}$. Here and below in this article it is assumed that $p^0=E_p$ (we are on the mass shell). For fixed $\vec{p}$, think of $A(\vec{p})$ and $B(\vec{p})$ as four--component columns whose entires are operators on a Hilbert space which has not yet been defined. The necessary and sufficient conditions for (\ref{psisol}) to provide a solution to (\ref{DEop}) are
\be
\begin{aligned}
&a)\quad (\not\!{p}-m)A(\vec{p})=0,\\
&b) \quad (\not\!{p}+m)B(\vec{p})=0 ,
\end{aligned}
\ee
where $\not\!{p}=p_a\gamma^a$. For fixed $\vec{p}$, the space of solutions for these equations are given by two--dimensional subspaces of $\mathbb{C}^4$ whose bases will be denoted by $u(\vec{p},\sigma)$ for $a)$ and $v(\vec{p},\sigma)$ for $b)$, where $\sigma$ numbers the two basis vectors and conventionally takes values $1/2$ and $-1/2$. The general solution of (\ref{DEop}) can thus be written as
\be\label{psiaac}
\widehat{\psi}(x)=\int\left({e^{-ip\cdot x}u_{\sigma}(\vec{p})a_{\sigma}(\vec{p})+e^{ip\cdot x}v_{\sigma}(\vec{p})a^{c\dag}_{\sigma}(\vec{p})}\right) d\Gamma_p ,
\ee
where $a_{\sigma}(\vec{p})$ and $a^{c\dag}_{\sigma}(\vec{p})$ are now completely arbitrary operator--valued functions of $\vec{p}$. Although these functions are arbitrary from the viewpoint of the Dirac equation, the anti--commutation relations for $\psi$, $\dot{\psi}$, $\overline{\psi}$ and $\dot{\overline{\psi}}$ that follow from the canonical formalism uniquely determine the anti--commutation relations between  $a_{\sigma}(\vec{p})$, $a^{c\dag}_{\sigma}(\vec{p})$ and their Hermitian conjugates. To see this, note that (\ref{psiaac}) can be inverted to yield
\be
\begin{aligned}
&u_{\sigma}(\vec{p})a_{\sigma}(\vec{p})=
\int e^{ip\cdot x}\left({E_p\widehat{\psi}(x)+i\dot{\widehat{\psi}}(x)}\right)d^3x \\
&v_{\sigma}(\vec{p})a^{c\dag}_{\sigma}(\vec{p})=
\int e^{-ip\cdot x}\left({E_p\widehat{\psi}(x)-i\dot{\widehat{\psi}}(x)}\right)d^3x \\
&u^*_{\sigma}(\vec{p})a^{\dag}_{\sigma}(\vec{p})=
\int e^{-ip\cdot x}\left({E_p\widehat{\psi}^{\dag}(x)-i\dot{\widehat{\psi}}^{\dag}(x)}\right)d^3x \\
&v^*_{\sigma}(\vec{p})a^{c}_{\sigma}(\vec{p})=
\int e^{ip\cdot x}\left({E_p\widehat{\psi}^{\dag}(x)+i\dot{\widehat{\psi}}^{\dag}(x)}\right)d^3x ,
\end{aligned}
\ee
where $a^c(\vec{p}):=\left({a^{c\dag}(\vec{p})}\right)^{\dag}$ and ${}^*$ means complex conjugation. Since the basis vectors $u_{\sigma}(\vec{p})$ for different values of $\sigma$ are linearly independent (the same concerns $v$, $u^*$ and $v^*$), these equations specify uniquely the coefficients 
 $a_{\sigma}(\vec{p})$, $a^{c\dag}_{\sigma}(\vec{p})$ and their conjugates in terms of field operators whose anti--commutation relations are known from the canonical analysis. One can seek for the choice of bases $u$ and $v$ which makes the anti--commutation relations between $a$'s particularly simple. For example, the choice proposed in \cite{Kazm5} leads to 
\be
\begin{aligned}
&\left[{a_{\sigma}(\vec{p}),a^{\dag}_{\sigma'}(\vec{p}\,')}\right]_+=
\left[{a^c_{\sigma}(\vec{p}),a^{c\dag}_{\sigma'}(\vec{p}\,')}\right]_+=
(2\pi)^32E_p\delta_{\sigma\sigma'}\delta(\vec{p}-\vec{p}\,'),\\
&\left[{a_{\sigma}(\vec{p}),a_{\sigma'}(\vec{p}\,')}\right]_+=
\left[{a^c_{\sigma}(\vec{p}),a^{c}_{\sigma'}(\vec{p}\,')}\right]_+=
\left[{a_{\sigma}(\vec{p}),a^{c\dag}_{\sigma'}(\vec{p}\,')}\right]_+=
\left[{a^c_{\sigma}(\vec{p}),a^{\dag}_{\sigma'}(\vec{p}\,')}\right]_+=0 .
\end{aligned}
\ee
These relations are so simple that they can be readily represented in the appropriate Hilbert space of many--particle states by the Fock quantization, which is described in all textbooks on quantum field theory. The commutation relations of $a$'s with the Hamiltonian support the interpretation of  $a^{\dag}_{\sigma}(\vec{p})$ and $a^{c\dag}_{\sigma}(\vec{p})$ as the operators that create particles of energy $E_p$ and their conjugates as annihilation operators. If the commutation relations with all the components of the energy--momentum tensor were investigated, the interpretation of the argument $\vec{p}$ as particle's momentum could be supported. The index $\sigma$ could be interpreted as the projection of particle's spin on the quantization axes in the particle's rest frame if the commutation relations of $a$'s with the components of spin density tensor were inspected. Finally, establishing the commutation relations with the electric charge operator allow for the interpretation of the superscript ${}^c$ as denoting anti--particles.

The reader might feel that the Poincar\'e symmetry that had been avoided so conscientiously throughout was finally employed. After all, both the energy--momentum tensor and the angular momentum tensor (including spin density part) are composed of Noether conserved currents that are related to the Poincar\'e symmetry of the action. I have just written that the commutation relations of creation and annihilation operators with momenta and spin operators are needed if the physical interpretation of $\vec{p}$ and $\sigma$ is to be established.

Although it is certainly true that the conserved quantities mentioned above can be obtained by Noether procedure from the Poincar\'e global symmetry, it does not mean that they can not be obtained differently. In general relativity, the energy momentum tensor of matter $t^{\mu\nu}$ is conventionally obtained by varying $\sqrt{|det(g)|}S_M$ with respect to the inverse space--time metric $g^{\mu\nu}$, were $S_M$ is the matter part of the action. This procedure is well defined in any space--time, no matter what symmetries it might posses, and yields the results that are compatible with Noether--Bellinfante method in flat space. There is a problem with spin density, though. If, however, gravity is interpreted as a Yang--Mills theory of the Poincar\'e group, than both the energy--momentum tensor and spin density tensor can be obtained by varying the matter action with respect to the tetrad and the connection (see \cite{Kazm1} for the details). Hence, the Poincar\'e symmetry of space--time is not really necessary. All the steps of the canonical quantization of the Dirac field performed in this article could be accomplished in principle in curved space as well. What is really problematic in curved space are all the technical complications. For example, it is not possible to find explicit solutions to the Dirac equation in most curved space--times.

\section{Conclusions}

The general canonical formalism described in Section \ref{intr} was successfully applied to the theory of the Dirac field in Section \ref{dirac}. It was argued that only after the Grassman odd variables are introduced the consistent canonical quantization of fermionic fields may be possible. Also, the Lagrangian for the Dirac field was shown to lead to the presence of constraints in the theory, which altogether made the application of a concept of generalized Dirac bracket necessary. It was argued that the Poincar\'e symmetry or the causal structure of space--time does not have to be involved in the program of quantization if the canonical method is consequently followed.

\section*{Acknowledgements}
I wish to thank W. Kaminski and J. Lewandowski for helpful comments.
This work was supported by grant N N202 287038.

\end{document}